\documentclass[12pt]{article}
\usepackage{setspace}
\usepackage{amsmath}
\usepackage{graphicx}
\usepackage{float}
\usepackage{verbatim}
\usepackage{amsthm}
\usepackage{amsfonts}
\usepackage{amssymb}
\usepackage{listings}
\usepackage{subcaption}
\usepackage{fullpage}

\newtheorem{mydef}{Definition}[section]
\newtheorem{lemma}{Lemma}[section]
\newtheorem{theorem}{Theorem}[section]

\newtheorem{claim}{Claim}[section]

\newcommand{\cE}{G}
\newcommand{\E}{{\mathbf E}}
\newcommand{\argmax}{{\text argmax~}}

\begin{document}
\title{Contracting Experts With Unknown Cost Structures}
\author{Mark Braverman\thanks{Princeton University, research partially supported by an Alfred P. Sloan Fellowship, an NSF CAREER award, and a Turing Centenary Fellowship. }\and
Gal Oshri \thanks{Princeton University}}

\maketitle

\begin{abstract}
We investigate the problem of a principal looking to contract an expert to provide a probability forecast for a categorical event.
We assume all experts have a common public prior on the event's probability, but can form more accurate opinions by engaging 
in research. Various experts' research costs are unknown to the principal. We present a truthful and efficient mechanism for the principal's problem of contracting an expert. This results in the principal contracting the best expert to do the work, and 
the principal's expected utility is equivalent to having the second best expert in-house. Our mechanism connects scoring rules with 
auctions, a connection that is useful when obtaining new information requires costly research. 

%We discuss several extensions to this mechanism. The contracts in \cite{Gneiting} are used to generalize our mechanism to non-binary events. We consider how the mechanism is affected when the principal and experts have a maximum acceptable risk and cannot afford to exceed a certain budget. Finally, we discuss the result of the experts changing their belief before the mechanism -- either due to a signal they received or due to costly research that they carry out.
\end{abstract}

\newpage

\setcounter{page}{1}

\section{Introduction}

This paper investigates the problem of a principal contracting an expert to provide a probability forecast on an event 
that will be observed in the future. 
 There are many situations where a party might be interested in predicting  a categorical event $\cE$ but is unable to research this itself. We call such a party the principal. This principal can contract an expert to do this research at some price. 
 If this price is fixed (without further transfers following the revelation of the outcome), the expert has no incentive to provide an accurate forecast. Instead, the principal pays the expert according to some predefined rule after observing the outcome of $\cE$. Our goal is to construct a mechanism that accounts for the expert attempting to maximize her own profit instead of the quality of information provided. Our contribution is a framework for contracting an expert with an unknown research cost structure in a way that aligns her incentives with the principal's incentives. We first discuss this problem more generally.

\subsection{Connections to Prior Work} 

Our work is an example of the principal-agent problem (or agency dilemma), which involves a principal who contracts an agent to perform a task. Both the principal and the agent aim to maximize their own utility. The question for the principal is how to align his own goals with the expert's self-directed goals. The problem is that the principal cannot observe the agent's actions directly to ensure the agent has fulfilled her obligations \cite{Eisenhardt}. Thus, we have a case of asymmetric information. This creates two problems: adverse selection and moral hazard. Adverse selection reflects the principal's ineffectiveness in determining whether the agent reveals her abilities truthfully before a deal is made. Moral hazard refers to the principal being unable to ensure the agent exerts maximum (or sufficient) effort after the deal has been made. Examples of the principal-agent problem include an `insurer who cannot observe the level of care taken by the person being insured' \cite{Grossman} and shareholders that cannot control the managers' actions \cite{Bebchuck}. In our case, the agent is an expert and the principal is a party that wants to learn information from the expert but cannot verify her information, as the expert does the research independently from the principal. A good mechanism needs to deal with both these problems by first ensuring that the best expert is chosen to do the job (where `best' is defined according to the principal's objective function), and then ensuring that once contracted the expert exerts an optimal level of effort in her research. 

This area of research has been applied extensively to probability forecasts. A simple example is a news channel, acting as the principal, that wants to give a forecast of the probability that it will rain the next day. They contract an expert (a meteorologist) who tells them the probability that it will rain. The news channel then pays the expert an amount depending on whether it rains or not (i.e. they observe the event before the expert's payment is determined). We cannot know whether the expert was `correct' as she provided a probability. However, we can expect that she will be paid more if, for example, she was quite certain that it would rain and it did rain than if it did not rain. Another example can be found in \cite{Olvera}, which discusses probability forecasts of inflation and GDP.

One approach to the problem of probability elicitation is the use of scoring rules, in which the expert reports a probability and is rewarded based on a function of the probability and the observed event. The main interest is in proper scoring rules, where we can show that the expert can do best (maximize her utility) by reporting the probability truthfully. There is a substantial body of work on proper scoring rules, see for example \cite{Gneiting}, \cite{Lichtendahl}, \cite{Winkler} and references therein. This has led to research into modifying proper scoring rules so that experts with nonlinear utility functions also make truthful reports. When the utility function of the expert is known, \cite{Winkler2} shows that a composite function of the utility function and a proper scoring rule is a proper scoring rule for the expert. \cite{Offerman} suggests a method that can be used when the utility function is not known. The method involves first learning the expert's preferences through experiments and then rewarding the expert with a proper scoring rule but correcting for biases suggested by the preferences found. Recently, Hossain and Okui designed and experimentally tested a method for constructing scoring rules that are incentive compatible regardless of the expert's risk-preference in \cite{Hossain}. They show that the new scoring rule elicits better probabilities than the quadratic scoring rule, which is one of the canonical proper scoring rules.

Another relevant area of research in probability elicitation is the study of prediction markets \cite{arrow2008promise}. These markets allow participants to buy and sell contracts in a form similar to traditional markets. For a binary event such as the outcome of the US elections, one form the contracts can take is the following: if a given candidate (say Obama) is elected, each contract pays \$1, and \$0 otherwise. The market price of these contracts, which is determined based on the buying and selling patterns of the participants, is an aggregation of the market's information. If many people buy such contracts, it suggests they believe the event will occur (at least with a probability higher than that determined by the current price of the contract). An analysis of the accuracy of prediction markets can be found in \cite{Berg}, which shows that they do well compared with, for example, political polls.  \cite{Hanson} discusses market scoring rules, which are similar to prediction markets in that they allow individuals to update the current aggregate belief but do not require a matching buyer/seller for each update. As most prediction markets are unsubsidized, the main weakness of predictive markets is that there appear to be no market equilibria where any of the experts does research.

While probabilistic forecasts from experts have been widely discussed in the literature, most of the previous research does not consider the cost the expert incurs in acquiring information or doing research. The use of a proper scoring rule incentivizes the expert to report her belief truthfully, but assumes that the expert already has such a belief. This assumption is unreasonable in many situations. For example, an expert in possession of an expensive medical test will not form an opinion about a patient's condition until the test is administered. It is also possible that it is more expensive for the expert to do higher quality research/testing. Since the expert acts as a self-interested agent, she will do more work only if that work will lead to gains that outweigh the higher cost. For example, if the medical expert can spend some additional amount of money to more precisely predict the condition, but that this will increase her payment by an amount lower than the cost, the expert will not do this. Clemen shows how to construct proper scoring rules that incentivize the expert to perform a number of trials chosen by the principal to achieve a certain level of precision in \cite{Clemen}. However, Clemen requires that the cost structure of the expert is known (i.e. how much it costs the expert to perform any number of trials). This assumption is restricting in many settings: an expert may consider her cost structure to be legally private information and may also lie about her costs to increase her profit. \cite{Carroll} finds the optimal contract in a model similar to ours, where a principal contracts an expert to acquire information and is unaware of the technology available to the expert. Our model is different in that we allow the principal to choose from several experts instead of being restricted to finding the best contract for a particular expert. Our contribution is a mechanism that does not require any assumptions about the experts' cost structures.

\subsection{Results}

%We investigate a probability forecast problem where the principal wants to find the probability that a certain event will occur. The principal can contract one of multiple experts to do research and report this probability, but the principal does not know the experts' cost structures. These cost structures are represented by a function from the different tests the experts can run to the cost associated with those tests. A test represents the probability distribution from which the expert's final belief is sampled. Therefore, the expert can choose between tests that are more or less likely to result in an informative belief (where being more certain that the event will occur or will not occur is seen as being more informative). 

Our starting point is a discrete categorical event $\cE$ with $n$ potential outcomes. The principal has a prior distribution $\rho_0$ 
on the outcomes of $\cE$, and would like to contract with an expert to acquire additional information. The principal has a utility function $P(\rho)$ for a posterior distribution $\rho$. Without loss of generality $P$ is convex in $\rho$ (i.e. the principal has 
a non-negative utility in expectation for additional information), and $P(\rho_0)=0$. There is a set of experts $e_1,\ldots, e_k$, 
who initially share the common prior $\rho_0$ (think for example of $e_1,\ldots,e_k$ as medical testing facilities, and of $\cE$ as 
the condition of a yet untested patient). Each expert $e_i$ has a set of testing technologies $M_i=\{\mu_1,\mu_2,...,\mu_k\}$ where each $\mu$ is a distribution on posterior distributions $\rho\sim\mu$ such that $\E_\mu \rho=\rho_0$ and $k$ is the number of tests available to the expert (we drop the index henceforth). Each $\mu$ has cost 
$C_i(\mu)$. We assume that if $\mu_0=1_{\rho_0}$ then $C_i(\mu_0)=0$, i.e. the expert always has the option of doing nothing at no cost. The optimal utility that expert $e_i$ would deliver if she were in-house (i.e. fully observable and controlled by the principal) would be $$ U_i:= \max_{\mu\in M_i} \left[ \E_{\rho\sim \mu} P(\rho)-C_i(\mu) \right].$$
If the principal where omniscient, he would pick the expert maximizing $U_i$. If he could also force the experts to reveal their efforts truthfully, he would be able to achieve utility $U_i$ to himself. We show that there is a truthful mechanism 
that allows the principal to pick the best expert, and to derive the utility $U'$ given by the second-best expert. In other words, 
the principal's utility corresponds to having the second-best expert in-house. 

%[MARK TO GAL: WE'LL NEED TO LINK THIS TO THE REST OF THE PAPER]

\begin{theorem} \label{thm:main_theorem}
There is a dominant strategy truthful mechanism which assigns the prediction job to expert $e_i$ such that $i=\argmax U_i$. Moreover, the principal's expected utility achieved by executing the mechanism is $U_P = \max_{j\neq i} U_j$. 
\end{theorem}

The mechanism is a second-price auction that involves the experts bidding on the best preference curve they can achieve and giving the top bidder the second-best preference curve. We prove that this mechanism is truthful: the experts maximize their utility by bidding the true best preference curve they can achieve. The main limitation of this mechanism is the inability to aggregate information from multiple experts -- something we will discuss briefly in the conclusions section.

\section{Problem Definition} \label{sec:problem_def}

Our goal is to design a mechanism that allows the principal to learn the outcome probabilities of $G$  (or rather, the expert's belief distribution $\rho$) by hiring an expert from a list of available experts. This mechanism should provide certain guarantees for both the principal and the expert in the form of properties of the mechanism. The first of these is the truthful property:

\begin{mydef} 
Truthful property: making truthful reports is a dominant strategy equilibrium for the experts.
\end{mydef}

This means that the expert will maximize her expected utility by making truthful reports regardless of the strategy employed by other players (i.e. experts). A truthful mechanism is useful because it does not require agents to think about what other strategies will maximize their utility. Furthermore, the principal can believe the expert's report instead of having to consider whether the expert lied. We also want the mechanism to fulfill the participation constraint:

\begin{mydef}
Participation constraint: the expected utilities of the principal and the experts cannot be less than zero.
\end{mydef}

The participation constraint requires that the agents can have a higher utility by participating instead of not participating (we set the utility of them not participating to be the baseline utility of zero). This constraint is important because if the expert would have a negative utility by participating in the mechanism, she would prefer not to join. Finally, we want the mechanism to be efficient:

\begin{mydef}
Efficient property: the expert contracted is the expert who can achieve the best principal preference curve (with a non-negative expected utility).
\end{mydef}

%As before, the principal has preference curves over the different probability vectors for event $G$. In the binary case, the utility for the expert for a probability $\rho$ was given by $P(\rho)$. Now, the utility is given by $P(\rho) = P(\rho_1, \rho_2, \ldots, \rho_n)$ which can be written as $P(\rho_1, \rho_2, \ldots, \rho_{n - 1})$ as we can remove the dependence on $\rho_n$ because $\rho_n = 1 - \rho_1 - \rho_2 - \ldots - \rho_{n - 1}$. This is why when $n = 2$ (the binary case where $G$ occurs or does not occur) our preference curves required only one scalar input parameter). As before, due to Lemma \ref{lem:convex_function}, $P(\rho)$ is convex. The preference curves can be translated along the $P$ direction to represent different utilities for the principal. We let $P_\beta(\rho)$ be the preference curve that has $P_\beta(\vec{\lambda}) = -\beta$.

\section{Principal Preference Curves} \label{sec:principal_preference_curves}

The principal is likely to have a trade-off between the probability distribution $\rho$ reported by the expert and the cost of acquiring this information. Thus, the principal would be willing to lose some precision in order to pay a smaller amount. To formalize this trade-off, we assume there is a price $P$ that the principal is willing to pay for every reported probability distribution such that the utility of the principal remains the same. This is the principal's utility function $P(\rho)$ which we take to be twice differentiable. Any point on the surface given by the utility function is equally preferred by the principal. Thus, function $P(\rho)$ is a preference curve (also known as an indifference curve). %See Figure \ref{fig:pref_curve_1} for an example of such a function where the preference curve is symmetric about the prior $\vec{\lambda} = (\frac{1}{2}, \frac{1}{2})$. 

%\begin{figure}[ht]
%	\centering
%		\includegraphics[scale=0.23]{pref_curve_1.png}
%	\caption{An example $P(\rho)$ preference curve. Every point on the curve has equal utility for the principal.}
%	\label{fig:pref_curve_1}
%\end{figure} 

There is one important property of the preference curve that we will use, which is convexity. We have:

\begin{lemma} \label{lem:convex_function}
$P(\rho)$ is a convex function of $\rho$.%\footnote{This means $t P(\rho') + (1 - t) P(\rho'') \geq P(t \rho' + (1 - t) \rho'')$ where $t \in [0, 1]$.}
\end{lemma}

\begin{proof}
The principal receives a report $\rho$ from the expert and has some set of actions he can take $\mathcal{A}$. Each action $a$ is a function from the set of possible outcomes ($\Sigma$) to a real payoff $a: \Sigma \rightarrow \mathbb{R}$. We can look at $\rho$ as a function of the outcomes ($\sigma$ represents an outcome) to the probability of each outcome occurring. We can now calculate the payoff for the principal when he receives a report $\rho = \rho(\sigma)$ as he chooses the action that maximizes his payoff:
$$P(\rho) = \max_{a \in \mathcal{A}} \sum_{\sigma \in \Sigma} \rho(\sigma) a(\sigma)$$
The expert chooses the action that maximizes her utility given the information she has. We want to show that $P(\rho)$ is convex: $t P(\rho') + (1 - t) P(\rho'') \geq P(t \rho' + (1 - t) \rho'')$. Let $\rho = t \rho' + (1 - t) \rho'')$. Let $a_\rho$ be the action that the principal chooses when receiving report $\rho$: $a_\rho = \operatorname*{arg\,max}_{a \in \mathcal{A}} \sum_{\sigma \in \Sigma} \rho(\sigma) a(\sigma)$. When the principal receives report $\rho$ and takes action $a_\rho$, he receives payoff $P(\rho) = P(\rho, a_\rho)$. If the principal received report $\rho'$ or $\rho''$ but still took action $a_\rho$, his payoff is at most the payoff he would get by choosing the action that maximizes the payoff (i.e. $P(\rho') \geq P(\rho', a_\rho)$ and $P(\rho'') \geq P(\rho'', a_\rho)$). Therefore, $t P(\rho') + (1 - t) P(\rho'') \geq t P(\rho', a_\rho) + (1 - t) P(\rho'', a_\rho) = P(\rho, a_\rho) = P(\rho) = P(t \rho' + (1 - t) \rho'')$, as required.
\end{proof}

The preference curve we defined can be shifted down in the price ($P$) direction, which is equivalent to the principal paying less for any given probability distribution $\rho$. The further down the function is shifted, the less the principal pays for each value of $\rho$, and the higher the utility for the principal will be. Different preference curves can be achieved by such shifts and each preference curve has a particular utility associated with it. The principal's goal then is to maximize this downward shift and therefore, to contract an expert to perform research corresponding to the best possible preference curve. We quantify the downward shift by the parameter $\beta$: we let $P_{\beta}(\rho)$ be the preference curve which has the point $P_{\beta}(\rho_0) = -\beta$. Thus, when $\beta = 0$, we have that $P_0(\rho_0) = 0$. This is the preference curve with 0 utility for the principal, so a payment of 0 means no gain or loss in utility from the ground state where no expert is contracted. The higher $\beta$ is, the higher the utility for the principal is (specifically, it is equal to $\beta$). We can thus say that a preference curve is better than another preference curve if its $\beta$ is higher. See Figure \ref{fig:pref_curve_2} for an example of several preference curves with different $\beta$ parameter values in the case where $G$ is a binary event ($n = 2$, so we can let $\rho$ in the figure be the probability of outcome 1 occurring). Note that the preference curves are not symmetric, symbolizing that the principal cares more about the accuracy of one of the outcomes. An example of this is when the principal wants to know if an airplane engine is broken or not. There isn't much difference between engine being broken with 80\% or 90\% probability (needs to be fixed). However, there is a big difference between 80\% and 90\% probability that the engine is not broken (precise level of certainty matters).

\begin{figure}[ht]
	\centering
		\includegraphics[scale=0.4]{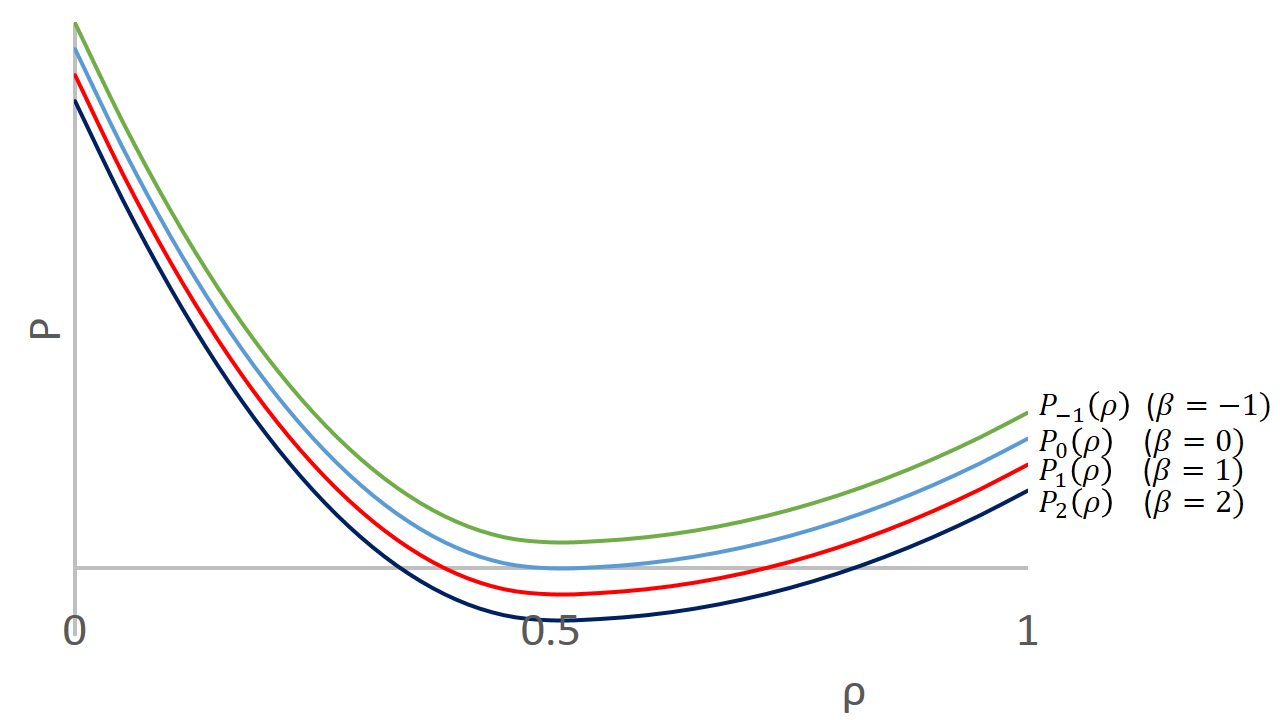}
	\caption{An example of $P_\beta(\rho)$ with several values of $\beta$ and asymmetric preference curves.}
	\label{fig:pref_curve_2}
\end{figure}

We now describe the contract that we use in our mechanism.

\section{Contract Description} \label{sec:contract_description}

A contract is an agreement between the principal and the expert regarding the information the expert provides the principal and how much the principal pays the expert (where this payment can depend on the information provided). Our goal is to design a contract which will force the expert to do research matching a chosen principal preference curve in order to maximize her expected utility. This means that the principal only pays the expert the amount corresponding to the $\rho$ the expert truly believes as given by the preference curve (a function from $\rho$ to a price). 

We want to find a contract where we can prove that given a contract corresponding to preference curve $P_\beta(\rho)$, the expert will report her belief $\rho^*$ truthfully and that this will lead to an expected payment from the principal to the expert of $P_\beta(\rho^*)$. We use a contract that corresponds to a proper scoring rule with an expected score equal to $P_\beta(\rho)$, as given  in \cite{Gneiting}. 

The contract corresponding to $P_\beta(\rho)$ is as follows:
\begin{enumerate}
\item Expert reports her belief $\rho^*$
\item Outcome $i$ is observed for event $G$
\item Expert is paid $P_\beta(\rho^*) - \left\langle {\nabla{P_\beta(\rho^*)} ,\rho^* } \right\rangle + \frac{\partial}{\partial\rho_i} P_\beta(\rho^*)$
\end{enumerate}

Where $\left\langle {\cdot , \cdot } \right\rangle$ is a scalar product. The expected payment to the expert is equal to $P_\beta(\rho^*)$ because that is the equivalent of function $G(p)$ in Theorem 2 of \cite{Gneiting} (the expected score function). Therefore:

\begin{lemma} \label{lem:generalized_corresponding_contract}
If the expert is given the contract corresponding to $P_\beta(\rho)$ and reports her belief $\rho^*$ truthfully, the expected payment from the principal to the expert is $P_\beta(\rho^*)$.
\end{lemma}

Furthermore, the payment of the contract, given by Theorem 2 of \cite{Gneiting}, is a proper scoring rule (the payment scheme to the expert is such that she maximizes her utility by making a truthful report). This gives us the following result:

\begin{lemma} \label{lem:generalized_truthful_contract}
An expert given a contract corresponding to the principal's preference curve characterized by $P_\beta(\rho)$ maximizes her expected profit by reporting the truthful value of her belief $\rho^*$.
\end{lemma}

It is instructive to understand why these lemmas are true from a geometric perspective. Suppose the expert reports $\rho^*$ for the contract corresponding to preference curve $P_\beta(\rho)$. We find the hyperplane tangent to the surface of $P_\beta(\rho)$ at $\rho$ (let this hyperplane be $L(\rho)$): 

\begin{align*}
L(\rho) &= P_\beta(\rho^*) + \left\langle {\nabla{P_\beta(\rho^*)} , (\rho - \rho^*) } \right\rangle  \\
&= P_\beta(\rho^*) - \left\langle {\nabla{P_\beta(\rho^*)} , \rho^* } \right\rangle + \left\langle {\nabla{P_\beta(\rho^*)} , \rho } \right\rangle
\end{align*}

Let $\hat{\rho_i}$ be the belief vector where $\rho_i = 1$ and $\rho_{j \neq i} = 0$. If outcome $i$ is observed, the expert is paid $L(\hat{\rho_i})$. This means the expert is paid the value of $L(\rho)$ when $\rho$ is the belief that is certain that outcome $i$ will occur. Therefore, the payment for outcome $i$ is $P_\beta(\rho^*) - \left\langle {\nabla{P_\beta(\rho^*)} , \rho^* } \right\rangle + \left\langle {\nabla{P_\beta(\rho^*)} , \hat{\rho_i} } \right\rangle = P_\beta(\rho^*) - \left\langle {\nabla{P_\beta(\rho^*)} , \rho^* } \right\rangle + \frac{\partial}{\partial\rho_i} P_\beta(\rho^*)$ as above.

We have shown how the contract was derived. Now we will discuss its properties. To see that Lemma \ref{lem:generalized_corresponding_contract} is true, we calculate the expected payment ($E[u]$) to the expert when she reports her true belief $\rho^*$:

\begin{align*}
E[u] &= \sum_{i = 1}^n (\text{probability that outcome } i \text{ occurs}) \cdot (\text{payment for outcome } i) \\
&= \sum_{i = 1}^n \rho_i^* \cdot (P_\beta(\rho^*) - \left\langle {\nabla{P_\beta(\rho^*)} ,\rho^* } \right\rangle  + \frac{\partial}{\partial\rho_i^*} P_\beta(\rho^*)) \\
&= P_\beta(\rho^*) - \left\langle {\nabla{P_\beta(\rho^*)} , \rho^* } \right\rangle + \sum_{i = 1}^n \rho_i^* \cdot \frac{\partial}{\partial\rho_i} P_\beta(\rho^*) \\
&= P_\beta(\rho^*) - \left\langle {\nabla{P_\beta(\rho^*)} , \rho^* } \right\rangle + \left\langle {\nabla{P_\beta(\rho^*)} , \rho^* } \right\rangle \\
&= P_\beta(\rho^*) 
\end{align*}

We now discuss the truthfulness of the contract. Suppose the expert has belief $\rho'$ but reports $\rho^*$. The expert's payment is based on the hyperplane constructed using $\rho^*$ (call this hyperplane $L^*(\rho)$. Her expected payment is $L^*(\rho')$. Since $P_\beta(\rho)$ is a convex function (from Lemma \ref{lem:convex_function}), the tangent plane $L(\rho)$ is a supporting hyperplane. That is, each point of $P_\beta(\rho)$ is above the hyperplane other than at $\rho^*$ where they are equal. This means that $L^*(\rho')$ is maximized when $\rho^* = \rho'$ because the hyperplane with the highest value at $\rho = \rho'$ is the one tangent to $P_\beta(\rho)$ at $\rho'$. This is the reasoning behind Lemma \ref{lem:generalized_truthful_contract}.

The next step is to consider the utility for the expert when given contracts corresponding to the different preference curves. The only difference between such contracts is the price of the contract: the payment from a contract corresponding to the preference curve $P_\beta(\rho)$ is equal to the payment from the contract corresponding to preference curve $P_0(\rho)$ minus $\beta$. This means that when the expert is choosing the research distribution $\mu$ to maximize her utility, she only needs to consider the utility for the contract corresponding to $P_0(\rho)$. Based on Lemma \ref{lem:generalized_corresponding_contract}, we know that if the expert reaches probability vector $\rho$, she will be paid $P_0(\rho)$. Thus, if the expert chooses research distribution $\mu$, her expected payment is $\E_{\rho\sim \mu} P_0(\rho)$
%, where we let $u = u(\mu)$ be the principal's utility:
%$$E[\text{payment}] = E[u(\mu)] = \int_{\Delta^{n - 1}} P_0(\rho) d\mu(\rho)$$

Expert $e_i$ will choose the $\mu$ that maximizes her expected utility. If $\mu'$ is the chosen $\mu$, then:
$$\mu' = \operatorname*{arg\,max}_\mu \left[ \E_{\rho\sim \mu} P_0(\rho) - C_i(\mu) \right]$$

This means the expert will choose the research distribution $\mu$ that maximizes her expected profit which is the utility gained from the contract minus the cost of doing that research. Thus, the maximum expected utility earned from the contract corresponding to $P_0(\rho)$ is, as we defined earlier:
$$U_i:= \max_{\mu\in M_i} \left[ \E_{\rho\sim \mu} P(\rho)-C_i(\mu) \right].$$

Therefore, expert $e_i$ will make an expected profit of 0 if she pays a fee $U_i$ for the contract corresponding to $P_0(\rho)$. This means the best preference curve the expert can achieve is the preference curve $P_\beta(\rho)$ where $\beta = U_i$. We thus have the following result:

\begin{lemma} \label{lem:final_best_pref_curve}
The best preference curve that an expert can achieve without making a negative expected profit is the one characterized by: $$\beta = U_i$$
\end{lemma}

We are now ready to present the mechanism leading to our main result. 

\section{Mechanism} \label{sec:mechanism}

Our mechanism involves running a second-price auction to choose which expert the principal contracts and what contract this expert receives. We have experts $e_i,\ldots,e_k$ who want to do the research, each with a cost function $C_i(\mu)$, where $i$ is the agent's index and their available research distributions $\mu$ and cost functions $C_i(\mu)$ are not necessarily identical. This is to account for the differences among the experts as they will differ in the technologies for research available to them. The mechanism is as follows:
\begin{enumerate}
\item The principal reveals $P_0(\rho)$ to the experts.
\item Expert $e_i$ submits a bid $b_i$. This is the value of $\beta$ in the contract which they will receive. We expect each expert's bid to be the highest value of $\beta$ for which the expert expects a non-negative profit (as we found in Lemma \ref{lem:final_best_pref_curve}).
\item Expert $e_i$ receives the contract corresponding to preference curve $P_\beta(\rho)$ with $\beta = b_j$, where experts $e_i$ and $e_j$ are the experts that submitted the highest and second highest bids respectively (we break ties randomly). The expert cannot resign the contract -- she must report a belief $\rho$ and be paid according to the contract.
\end{enumerate}

We are now ready to prove Theorem \ref{thm:main_theorem}:

%\begin{theorem} \label{thm:truthful_mechanism}
%Truthful bidding is a dominant-strategy equilibrium in the preference curve mechanism and the mechanism is also efficient in this equilibrium.
%\end{theorem}

\begin{proof} We follow the canonical second-price auction truthfulness proof as in \cite{Parkes}. Consider the available experts and fix attention on expert $e_1$. Let $v_1$ be the $\beta$ value of the best preference curve the expert can achieve without making negative expected profit. Let $b' = \max_{j \neq 1}b_j$ denote the maximum bid from the other experts. We proceed by case analysis.

Case $1$: $v_1 > b'$. Expert $e_1$'s best response is to bid greater than $b'$ because he will then receive the contract with $\beta = b'$ from which he can make an expected profit of $v_1 - b'$. In particular, bidding $b_1 = v_1$ is a best response.

Case $2$: $v_1 = b'$. Expert $e_1$ is indifferent across all bids. If he wins, he receives a contract from which he makes an expected profit of $0$. If he loses, his utility remains $0$. In particular, bidding $b_1 = v_1$ is a best response.

Case $3$: $v_1 < b'$. Expert $e_1$'s best response is to bid less than $b'$. This will generate a utility $0$, whereas bidding higher than $b'$ will result in winning a contract from which the expert will make an expected loss, leading to a utility below $0$. In particular, bidding $b_1 = v_1$ is a best response.

We conclude that regardless of what the other experts bid, expert $e_1$ should bid truthfully. Since we chose expert $e_1$ without a loss in generality, this applies to all the experts. Therefore, truthful bidding is a dominant-strategy equilibrium. 

We also see that the contract is given to the expert with the highest bid, who is the expert that can achieve the best principal's preference curve (without an expected loss) out of all the experts (i.e. expert $e_i$ such that $i = \argmax U_i$. Therefore, the mechanism is efficient.

Finally, since the contract chosen is the one corresponding to preference curve with $\beta = b_j$ where $j =\argmax b_j =\argmax U_j$ with $j \neq i$, we have that $U_P = \max_{j\neq i} U_j$. 
\end{proof}

The mechanism also fulfills the participation constraint. As long as experts make truthful bids, they will receive contracts from which they can make non-negative expected utility (as long as their probability reports are also truthful). 

Furthermore, the mechanism allows the principal to have a lower bound preference curve (corresponding to a second-price auction with reserve). Consider the situation where the principal has some preference curve that is the worst preference curve he would want to contract an expert for. This means he has a minimum value of $\beta'$ for the preference curves $P_\beta(\rho)$ that he wants to achieve. Even if he can find an expert who will take the contract matching the preference curve with a $\beta < \beta'$, the principal would prefer not to contract the expert. This is useful in a realistic setting, where the principal has minimum requirements on what precision level should be learned from any particular cost of research. 

The solution for both cases is to run the same mechanism as above but as a second-price auction with a reserve. This means there is a virtual bid corresponding to the minimum preference curve $\beta'$. The auction proceeds as before, but if this virtual bid is the highest, no expert is contracted. This preserves the properties of the second-price auction but guarantees to the principal that an expert will only be contracted if a preference curve $P_\beta(\rho)$ with $\beta \geq \beta'$ can be achieved.

%The intuition behind our approach is that the principal runs a second-price auction where experts bid to be contracted. The principal wants to contract the best expert (the one who can perform the required level of research at the lowest cost) for the cost of the second best expert, which intuitively is the best we can do considering the fact that the principal does not know the cost functions of the experts. The reasoning for this is that we can reduce our mechanism to a single item auction where the principal is the seller and the experts are the buyers. We will see that the contracts offered differ only in their fixed price. If we consider the contract with a fixed price of 0, every expert has some private valuation for this contract (the single item of the auction). This valuation depends on the expert's cost structure: an expert that has lower costs of research will have a higher valuation. Thus, the reduction is complete and by Myerson's result in \cite{Myerson} we see that this is the best we can do.

We can also think about our result in terms of in-house experts. An in-house expert is one that works directly for the principal (they share costs). The principal knows what research distributions $\mu$ are available to the expert and also their cost $C(\mu)$. Thus, the principal knows what the best preference curve the expert can achieve is and give him a contract corresponding to that preference curve. The cost incurred by the principal for the research is equal to the cost the principal would pay an outside expert who is given the contract. With this view of in-house experts, our result can also be stated as:

\begin{theorem}
There exists a truthful mechanism for hiring the best expert to do work equivalent to having the second best expert in-house.
\end{theorem}

\section{Maximum Risk} %\label{sec:max_risk}
Unless the expert is completely certain about which outcome of $G$ will occur, she has some uncertainty about her exact payment. For example, in the binary case, the expert might do a lot of costly research and believe that $G$ will occur with a high probability. This can lead to a contract that pays her if $G$ occurs but requires her to pay if $G$ does not occur. If $G$ does not occur, she loses the money spent on research and the amount she has to pay. In a realistic setting, she may have limited funds which she cannot exceed. This can be viewed as a maximum risk or the need for limited liability.

As an example, let $G$ be a binary event and that the expert believes $G$ will occur with probability 0.9. Consider the two contracts described in the following table:
\par
\bigskip

\begin{center}
\begin{tabular}{| l | c | c | r }
  \hline                        
   & Payment when $G$ does not occur & Payment when $G$ occurs \\ \hline
  Contract A & -8,000 & 1,000 \\ \hline
  Contract B & -8,000,000 & 889,000 \\
  \hline  
\end{tabular}
\end{center}
\par
\bigskip

Both contracts have an expected payment of 100 for the expert. An expert might be willing to risk 8,000 to make an expected profit of 100 with contract A, but she might not be able to pay 8,000,000 if she takes contract B and $G$ does not occur. Therefore, even though our previous analysis viewed contracts A and B equally, realistically they are not equivalent. Note that the principal might also have a maximum amount he can pay the expert for the information.

We are interested in how the mechanism has to be modified when there is a maximum amount the experts and principal are willing to pay. We let $\phi_p$, $\phi_e$ be the maximum amounts the principal and expert are willing to pay respectively and then derive how using these limits affect the contracts that can be used. We state the result here and provide the derivation (and an example) in Appendix section \ref{appendix:max_risk}. Our result limits how high the expert can bid as the research tests she can run ($\mu$) are now restricted. The new bid by the expert is given by the following lemma:

\begin{lemma} \label{lem:max_risk_lemma}
The best preference curve that an expert can achieve without making a negative expected profit and violating her own or the principal's maximum risk is the one characterized by:
$$\beta' = \operatorname*{arg\,max}_\beta \left[\max_{\mu \in M(\beta)} \int_{\Delta^{n - 1}} P_\beta(\rho) \,d\mu(\rho) - C(\mu) \right]$$
\end{lemma}

Where

$$M(\beta) = \{\mu \,|\, \text{If } \mu(\rho) > 0 \text{ then } -\phi_e \leq P_\beta(\rho) - \\
\left\langle {\nabla{P_\beta(\rho)} ,\rho } \right\rangle + \frac{\partial}{\partial\rho_i} P_\beta(\rho) \leq \phi_p \text{ for } 1 \leq i \leq n\}$$

\section{Discussion} \label{sec:discussion}

One question we might ask is whether we can change the contract while maintaining the important properties it has. For example, to deal with the maximum risk problem discussed above, we might try a different solution of simply changing the payment of each outcome to an acceptable range $[-\phi_e, \phi_p]$. However, we cannot do this without losing important properties of the contract. We have the following result:

\begin{theorem}
Our contract is unique in its equality in expectation of payment to the preference curve (Lemma \ref{lem:generalized_corresponding_contract}) and its truthfulness property (Lemma \ref{lem:generalized_truthful_contract}).
\end{theorem}

\begin{proof}
A contract needs to specify the payments made from the principal to the expert for each outcome of $G$. Suppose there is a contract other than the one given by our mechanism ($L(\rho)$) that fulfills Lemmas \ref{lem:generalized_corresponding_contract} and \ref{lem:generalized_truthful_contract}. This contract pays $\alpha_i$ if outcome $i$ occurs. For Lemma \ref{lem:generalized_corresponding_contract} to be fulfilled, when the expert believes $\rho^*$ we need her expected payment to be $P_\beta(\rho)$:
$$\sum_{i = 1}^n \rho_i^* \cdot \alpha_i = P_\beta(\rho^*)$$
This is the equation for a hyperplane and thus we know the contract can be represented as a hyperplane $L'(\rho)$. We know that $L(\rho^*) = L'(\rho^*) = P_\beta(\rho^*)$. For the two hyperplanes to be different, $L'(\rho)$ must be a rotation of $L(\rho)$ about its value at $\rho^*$. We want to know whether $L'(\rho)$ is also a supporting hyperplane.

The preference curve $P_\beta(\rho)$ is twice differentiable, which requires that it is differentiable at each $\rho$. By Theorem 3.1 of \cite{Berkovitz}, $P_\beta(\rho)$ is differentiable at $\rho$ if and only if $P_\beta(\rho)$ has a unique subgradient at $\rho$ (this unique subgradient is $\nabla P_\beta(\rho)$) because $P_\beta(\rho)$ is a convex function. Therefore, $L(\rho)$ is the unique supporting hyperplane and $L'(\rho)$ is not. This means that $L'(\rho)$ can have values above $P_\beta(\rho)$. Suppose this occurs, for example, at $\rho'$ ($P_\beta(\rho') < L'(\rho')$). If the expert has belief $\rho'$ but reports $\rho^*$, her expected payment increases from $P_\beta(\rho')$ to $L'(\rho')$. Thus, it is not optimal for the expert to be truthful. Therefore, the only contract that fulfills Lemmas \ref{lem:generalized_corresponding_contract} and \ref{lem:generalized_truthful_contract} is the one we use.
\end{proof}

We will now discuss the connection of our mechanism to auctions. Our mechanism for a single item can be reduced to a second-price auction and by Myerson's result in \cite{Myerson} is the best we can do (when we include a reserve). The reduction is as follows. The principal is the seller and the experts are the buyers. The expected payment an expert can receive from the contract after she pays for it is the value of the auction item for the expert. The contracts given by the principal differ only in their $\beta$ value. Since this is how much the expert would pay to receive the contract, we can view the $\beta$ that an expert bids to be her valuation. Thus, the contracts are the items used in the second-price auction. Therefore:

\begin{claim}
Our mechanism for a single contract can be reduced to a second-price auction.
\end{claim}

\section{Future Work} \label{sec:future_work}

% Risk neutral
%There are many interesting directions that future work on this topic can take, and we offer a few examples. These ideas mainly involve removing assumptions in the model to make it more applicable to the real world. We can first consider our assumption that the experts are risk-neutral (have a linear function from money to utility). If the experts are risk-averse or risk-seeking instead (have a non-linear utility function), the mechanism may no longer be truthful. This is due to similar reasons as in \cite{Winkler}. For example, a risk-seeking expert might gain from hedging his answer towards certainty because, compared with a risk-neutral agent, he gains more utility from being right while losing less utility from being wrong. If we know the expert's utility function, we might be able to modify the contract to maintain truthfulness, as \cite{Winkler2} does for proper scoring rules. If the utility function is not known, perhaps an approach similar to that used in \cite{Offerman} could be used. The problem is that the main interest in previous research is truthfulness whereas we also require that the expected payment match how much the principal is willing to pay for the given result. Future research could explore how to modify the mechanism and contracts to account for non-linear utility functions.

We have already explored some variation in the utility functions when we discussed the case of maximum risk, where the principal or the experts may have a limited budget that they cannot exceed, in Section \ref{sec:discussion}. This can be represented by a utility function that goes to $-\infty$ when this budget is exceeded (but is linear elsewhere as per our assumption). Our discussion assigned an absolute value to the budget of the parties. However, their preferences might be more complex: an expert might be willing to risk a certain amount of money only for a certain expected payout (for example, willing to risk 1,000,000 if she will win 10,000,000 with at least 0.95 probability). Forecast elicitation with agents who are not risk-neutral has been discussed in, for example, \cite{Winkler}, \cite{Winkler2}, \cite{Offerman}. Investigating realistic risk preferences and applying them to the mechanism is an interesting path for future work to take.

Our mechanism allows the principal to give a contract to only one expert. However, in some applications, it might be beneficial to get multiple experts to cooperate. Two experts can collaborate to decrease test costs, which would benefit the principal. This will require new types of contracts that incentivize the experts to work together based on their joint information structure. One very natural approach inspired by our one-expert mechanism is to try to lift results about multi-unit auctions to multi-expert mechanisms (where each ``item" correspond to the right to a prediction contract). The difficulty with this approach is that the principal's utility may be a very complex function of the experts' signals. Moreover, the marginal utility of one expert's technology may depend on the technology available to the other expert. In general, this means that there is little hope of approximating e.g. having two of the top experts in house using a mechanism such as ours. By adding assumptions on the information structure (such as independence) one can expect to make progress in this setting. 

The problem we explored involves a principal who contracted an expert for a one-time prediction. If the principal was to contract experts for predictions on a continual and sequential basis, he could learn about the performance of experts and use this information to increase his utility. This changes the problem to one of online learning. It would be interesting to explore how the principal can use this new information and how much it increases his utility. The experts might also be able to use this information. %In Section \ref{sec:research_before_bidding} we showed that not doing research before bidding is not an equilibrium. The experts can use information they learn about other experts (such as whether they do research before bidding) to choose their own actions.

\newpage

\bibliographystyle{acm}
\bibliography{bibliography}

\newpage

\appendix

\section*{Appendix}

\section{Maximum Risk Derivation and Example} \label{appendix:max_risk}

We are interested in finding how such maximum risk for the principal or expert might affect the results of the mechanism. We achieve this using our geometrical understanding of the contracts from our mechanism. First, we consider what happens if the principal has some maximum amount $\phi_p$ that he can afford to pay. Suppose a contract exists that has a payment greater than $\phi_p$ for some outcome. This requires that the supporting hyperplane of some point on the preference curve $P_\beta(\rho)$ has a value greater than $\phi_p$ for input $\hat{\rho_i}$ where $1 \leq i \leq n$, as this means the principal pays the expert more than $\phi_p$ for outcome $i$. However, since the supporting hyperplane can never be above the convex function, this only occurs if the preference curve has a value above $\phi_p$. Since the principal knows his preference curves, he can ensure no contract is given that might result in a payment greater than $\phi_p$. One approach to this is:
\begin{enumerate}
\item Calculate the minimum $\beta'$ such that $P_{\beta'}(\hat{\rho_i}) \leq \phi_p$ for $1 \leq i \leq n$.
\item Run the mechanism as a second-price auction with reserve $\beta'$.
\end{enumerate}  

While this approach is simple (requires just one calculation from the principal and no additional work from the experts), it might be restrictive. It is possible that none of the experts can achieve even the minimum preference curve required by the auction, but they would be able to achieve a point on a worse preference curve that does not require the principal to pay more than $\phi_p$. Both the principal and the expert lose the possibility for positive expected utility because of this. 

Another approach could change the contracts instead of the auction. The principal can announce that he will not pay more than $\phi_p$. This places a restriction on which probability reports $(\rho^*)$ can be made for each contract. This restriction for the contract corresponding to $P_{\beta'}(\rho)$ is so that only reports $\rho^* \in \{\rho^* \,|\, P_\beta(\rho^*) - \left\langle {\nabla{P_\beta(\rho^*)} ,\rho^* } \right\rangle + \frac{\partial}{\partial\rho_i} P_\beta(\rho^*) \leq \phi_p \text{ for } 1 \leq i \leq n\}$ can be made. Geometrically, this means we don't allow probability reports that lead to a supporting hyperplane that has a value greater than $\phi_p$ for any outcome.

This approach will also work when the expert has a maximum risk of $\phi_e$. If the expert can only pay at most $\phi_e$, the restriction now is that the expert can only report $\rho^* \in \{\rho^* \,|\, P_\beta(\rho^*) - \left\langle {\nabla{P_\beta(\rho^*)} ,\rho^* } \right\rangle + \frac{\partial}{\partial\rho_i} P_\beta(\rho^*) \geq -\phi_e \text{ for } 1 \leq i \leq n\}$. Before reporting $\rho^*$, the expert can check that there is no possibility that she will be forced to pay more than $\phi_e$ by fulfilling the above requirement. 

We have described how the principal and experts verify that the contract does not require them to pay more than is acceptable. For each contract, we restrict the domain so that allowable reports are: 
$$\rho^* \in \{\rho^* \,|\, -\phi_e \leq P_\beta(\rho^*) - \left\langle {\nabla{P_\beta(\rho^*)} ,\rho^* } \right\rangle \\
+ \frac{\partial}{\partial\rho_i} P_\beta(\rho^*) \leq \phi_p \text{ for } 1 \leq i \leq n\}$$

However, the expert's optimization problem has now changed because of this limitation. The research tests $\mu$ we have discussed can lead to different values of $\rho$, some of which will not fulfill the above criteria (and this criteria depends on the $\beta$ of the preference curve chosen). We are interested in finding the highest $\beta$ for which the expert makes a non-negative expected profit with the additional constraint above. We let $M(\beta)$ be a set of research tests $\mu$ that is a function of $\beta$ (i.e. the research tests in the set depend on the input parameter $\beta$). This set represents the research tests that lead to reports $\rho$ that fulfill the restrictions:

$$M(\beta) = \{\mu \,|\, \text{If } \mu(\rho) > 0 \text{ then } -\phi_e \leq P_\beta(\rho) \\
- \left\langle {\nabla{P_\beta(\rho)} ,\rho } \right\rangle + \frac{\partial}{\partial\rho_i} P_\beta(\rho) \leq \phi_p \text{ for } 1 \leq i \leq n\}$$

We now have the optimization problem the expert needs to solve to find the best preference curve she can achieve (and thus her optimal bid value) that we stated in Lemma \ref{lem:max_risk_lemma}

It will be helpful to view an example of this for the binary event case ($n=2$, $\rho$ and $\rho_0$ are scalar values for outcome 1). Suppose the principal has no maximum risk ($\phi_p = \infty$) but the expert we are considering has maximum risk $\phi_e$ (and $-\phi_e < P_\beta(\phi_0)$ as the expert won't bid more than her maximum risk). We will show how to compute the probability reports that the expert can make for the contract corresponding to $P_\beta(\rho)$. The expert can only carry out research tests $\mu$ that have a positive probability of such probability reports. The set of tests allowed is $M(\beta)$. If the expert reports $\rho^*$, the payments for when $G$ occurs and does not occur are chosen by drawing the tangent line to $P_\beta(\rho)$ at $\rho^*$ and finding the intersections with the lines $\rho = 0$ and $\rho = 1$ (which correspond to the two different outcomes). Since $P_\beta(\rho)$ is convex and therefore has non-decreasing gradient, we can calculate the range of allowed $\rho^*$ by finding the minimum and maximum values allowed as it is a continuous range in our example. To find the maximum allowed $\rho^*$ ($\rho_{max}^*$), we find the straight line (which represents the contract) that is tangent to the preference curve at $\rho_{max}^*$ (has gradient $P_\beta'(\rho_{max}^*)$) and intersects $\rho = 0$ at $-\phi_e$:

$$-\phi_e - P_\beta(\rho_{max}^*) = P_\beta'(\rho_{max}^*) \cdot (0 - \rho_{max}^*)$$
Similarly, to find the minimum allowed $\rho^*$ ($\rho_{min}^*$) we find the tangent that intersects $\rho = 1$ at $-\phi_e$:
$$-\phi_e - P_\beta(\rho_{min}^*) = P_\beta'(\rho_{min}^*) \cdot (1 - \rho_{min}^*)$$

Since $P_\beta(\rho) = P_0(\rho) - \beta$ and $P_\beta'(\rho) = P_0'(\rho)$, we can change the above two equations to be variables in $\beta$ and $\rho_{max}^*$ or $\rho_{min}^*$. This will allow the expert to calculate the allowed range of $\rho^*$ (which is within $[0, 1]$) as well as analyze how changing $\beta$ affects this range:
$$-\phi_e - P_0(\rho_{max}^*) + \beta = -P_0'(\rho_{max}^*) \cdot \rho_{max}^*$$
$$-\phi_e - P_0(\rho_{min}^*) + \beta = P_0'(\rho_{min}^*) \cdot (1 - \rho_{min}^*)$$

We can view this geometrically by drawing the lines that are tangents to $P_\beta(\rho)$ that pass through $-\phi_e$ at $\rho = 0$ and $\rho = 1$ as we do in Figure \ref{fig:max_risk}. The research tests $\mu$ that are allowed are ones that are restricted to $\rho_{min} \leq \rho^* \leq \rho_{max}$. In this example, since the expected payment for all $\rho^*$ in this range is negative, we know that the expert cannot earn a positive expected payment from this contract. This is not dependent on the expert's cost structure.

\begin{figure}[ht]
	\centering
		\includegraphics[scale=0.4]{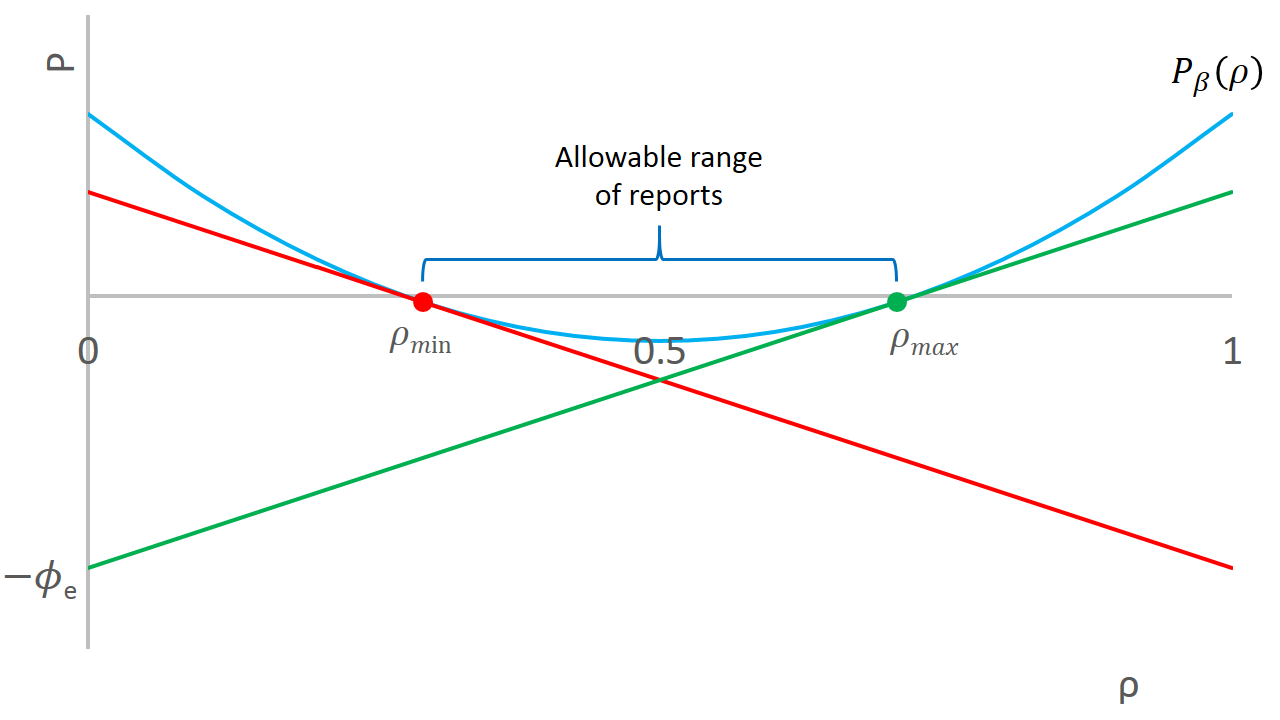}
	\caption{An example of finding the range of $\rho^*$ that the expert can report when she has a limited budget of $\phi_e$ and $G$ is a binary event.}
	\label{fig:max_risk}
\end{figure}

\end{document}